\documentclass[11pt,a4paper]{article}
\usepackage[T1]{fontenc}
\usepackage[utf8]{inputenc}
\usepackage[margin=1in]{geometry}
\usepackage{mathptmx,courier}
\usepackage{amsmath,amsfonts}
\usepackage{graphicx}
\usepackage[font=small,labelfont=bf]{caption}
\usepackage{subfigure}
\usepackage{booktabs}
\usepackage{comment}
\usepackage{amsfonts,amsmath,amsthm}
\usepackage{microtype}
\usepackage[numbers,sort&compress]{natbib}

\usepackage[hyphens]{url}
\usepackage[hidelinks]{hyperref}

\theoremstyle{plain}
\newtheorem{proposition}{Proposition}
\newtheorem{corollary}{Corollary}

\theoremstyle{remark}
\newtheorem{remark}{Remark}

\hypersetup{
  pdftitle={Breaking Behavioral Uniformity in Traffic Flow: Analytical Derivation and Numerical Verification of Optimal Cooperative Buffering},
  pdfauthor={Raphael Korbmacher and Antoine Tordeux},
  pdfkeywords={traffic flow, cooperative driving, string stability, car-following, optimization}
}
\newcommand{\keywords}[1]{\par\smallskip\noindent\begingroup\small
  \def\and{; }\textbf{Keywords:} #1\par\endgroup\medskip}
\newenvironment{acknowledgements}{\section*{Acknowledgements}}{}

\title{Breaking Behavioral Uniformity in Traffic Flow:\\
Analytical Derivation and Numerical Verification of Optimal Cooperative Buffering}
\author{Raphael Korbmacher \quad Antoine Tordeux\\[6pt]
\small School of Mechanical Engineering and Safety Engineering\\
\small University of Wuppertal, 42119 Wuppertal, Germany\\
\small\href{mailto:korbmacher@uni-wuppertal.de}{korbmacher@uni-wuppertal.de}}
\date{}
\begin{document}
\maketitle

\begin{abstract}

Stop-and-go waves pose a persistent challenge in traffic flow theory, significantly compromising highway safety, operational efficiency, and environmental sustainability. This paper mathematically derives an optimal cooperative driving policy based on the premise that connected vehicles can cooperate and assume distinct roles. In doing so, we depart from the classical assumption of uniform vehicular behavior. Through a rigorous optimization formulation constrained by linear stability conditions, it is analytically proved that the traffic system's throughput is globally maximized when a single designated "buffer" vehicle maintains an enlarged headway to absorb perturbations, enabling the remaining vehicles to form a tightly spaced, high-capacity platoon. This analytically derived optimum is confirmed via numerical spectral analysis and stability-constrained flow optimization using the Intelligent Driver Model and the Full Velocity Difference Model. The results demonstrate that this non-uniform cooperative buffering strategy yields dual benefits. It radically expands the string-stable parameter space to actively dissipate SGWs, while simultaneously achieving significantly higher traffic flow compared to both unbuffered traffic and non-cooperative Jam-Absorption Driving strategies.
\end{abstract}


\section{Introduction}
Traffic jams represent an undesirable traffic condition with numerous negative consequences, including economic losses, increased fuel consumption, and higher pollutant emissions~\cite{nishi2013theory,li2014stop}. While congestion can be triggered by explicit physical bottlenecks, accidents, or construction sites, it frequently emerges spontaneously in dense traffic for no apparent reason. These spontaneous breakdowns are commonly referred to as phantom jams or stop-and-go waves (SGWs). Despite over seven decades of research into traffic flow dynamics, the dissipation of these self-sustaining wave formations remains a complex and persistent endeavor.

Strategies to mitigate SGWs have traditionally been pursued along two complementary lines~\cite{nishi2013theory}. Macroscopic strategies utilize traffic management tools, such as variable speed limits, to regulate aggregate traffic flow~\cite{hegyi2005optimal,hegyi2010dynamic}. In contrast, microscopic approaches focus on how individual driving behaviors influence the formation or dissipation of SGWs. The collective dynamics of these microscopic interactions are typically evaluated through traffic string stability analysis, which formally defines traffic as unstable if microscopic velocity perturbations amplify as they propagate upstream through a vehicle platoon~\cite{korbmacher2025understanding}.

That human-driven traffic is indeed string-unstable has been demonstrated in a series of controlled experiments. The seminal circuit-road study by Sugiyama et al.~\cite{sugiyama2008traffic} showed that vehicles on a closed-loop track spontaneously transition into SGW due to small fluctuations in driver behavior, even without bottlenecks or external disruptions. Subsequent experiments have consistently confirmed this fundamental instability~\cite{tadaki2013phase,jiang2014traffic,stern2018dissipation}. One microscopic approach to counteract these instabilities is Jam-Absorption Driving (JAD) strategy, in which individual drivers maintain enlarged time gaps or adopt ``slow-in, fast-out'' maneuvers to absorb upcoming waves~\cite{nishi2013theory,wang2021jam,he2025review}.

Beyond behavioral interventions, considerable expectations have been placed on technological solutions, particularly Adaptive Cruise Control (ACC) and autonomous vehicles~\cite{treiber2001microsimulations,davis2004effect}. Yet empirical investigations have revealed that commercially deployed ACC systems are themselves string unstable~\cite{gunter2020commercially,makridis2020empirical,ciuffo2021requiem}. Researchers broadly agree that vehicle-to-vehicle (V2V) communication will be necessary to overcome these limitations~\cite{milanes2013cooperative,de2004design}.
Cooperative Adaptive Cruise Control (CACC) systems leverage such communication to reduce effective response times~\cite{ploeg2013controller,zheng2015stability,orosz2016connected}. However, even CACC strategies share the common structural assumption that all connected vehicles employ the same controller with identical parameters, achieving stability through richer information exchange rather than through deliberate cooperation in the control policy itself.

\begin{figure}[tbp]
    \centering
    \begin{minipage}{0.47\textwidth}
        \centering
        \subfigure[]{\includegraphics[width=\textwidth]{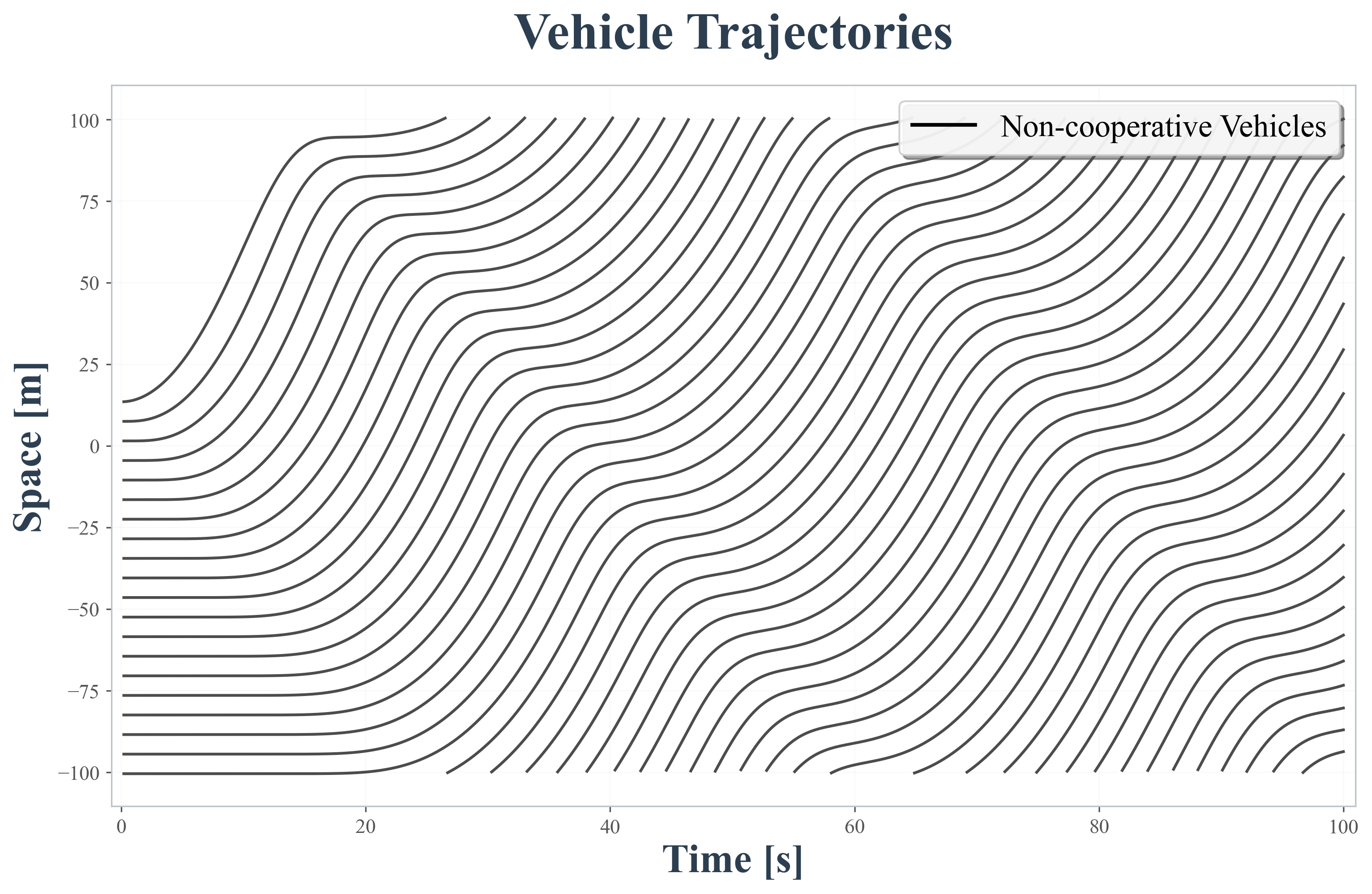}} \\ \vspace{1mm}
        \subfigure[]{\includegraphics[width=\textwidth]{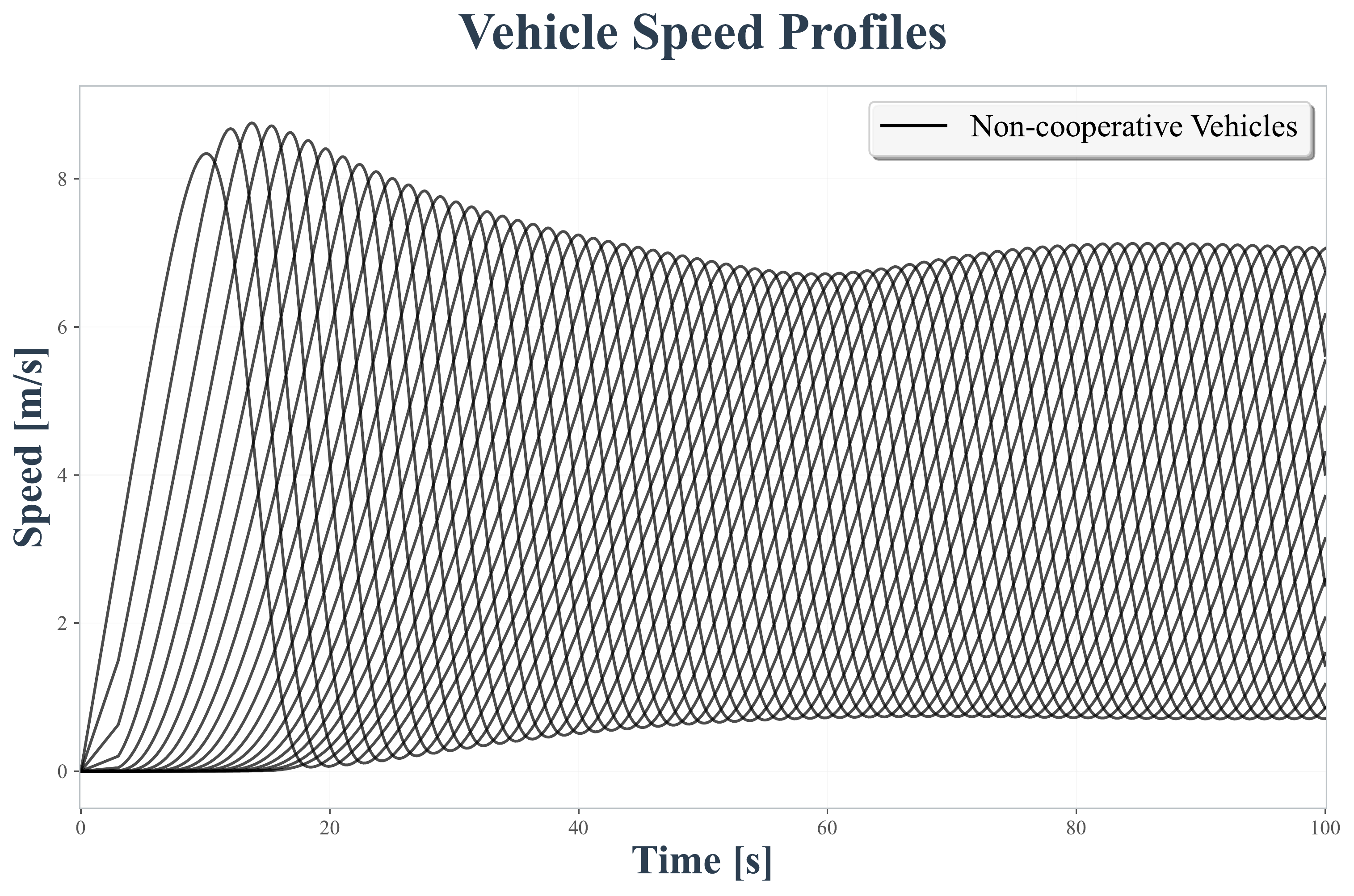}}
    \end{minipage}
    \hfill
    \begin{minipage}{0.47\textwidth}
        \centering
        \subfigure[]{\includegraphics[width=\textwidth]{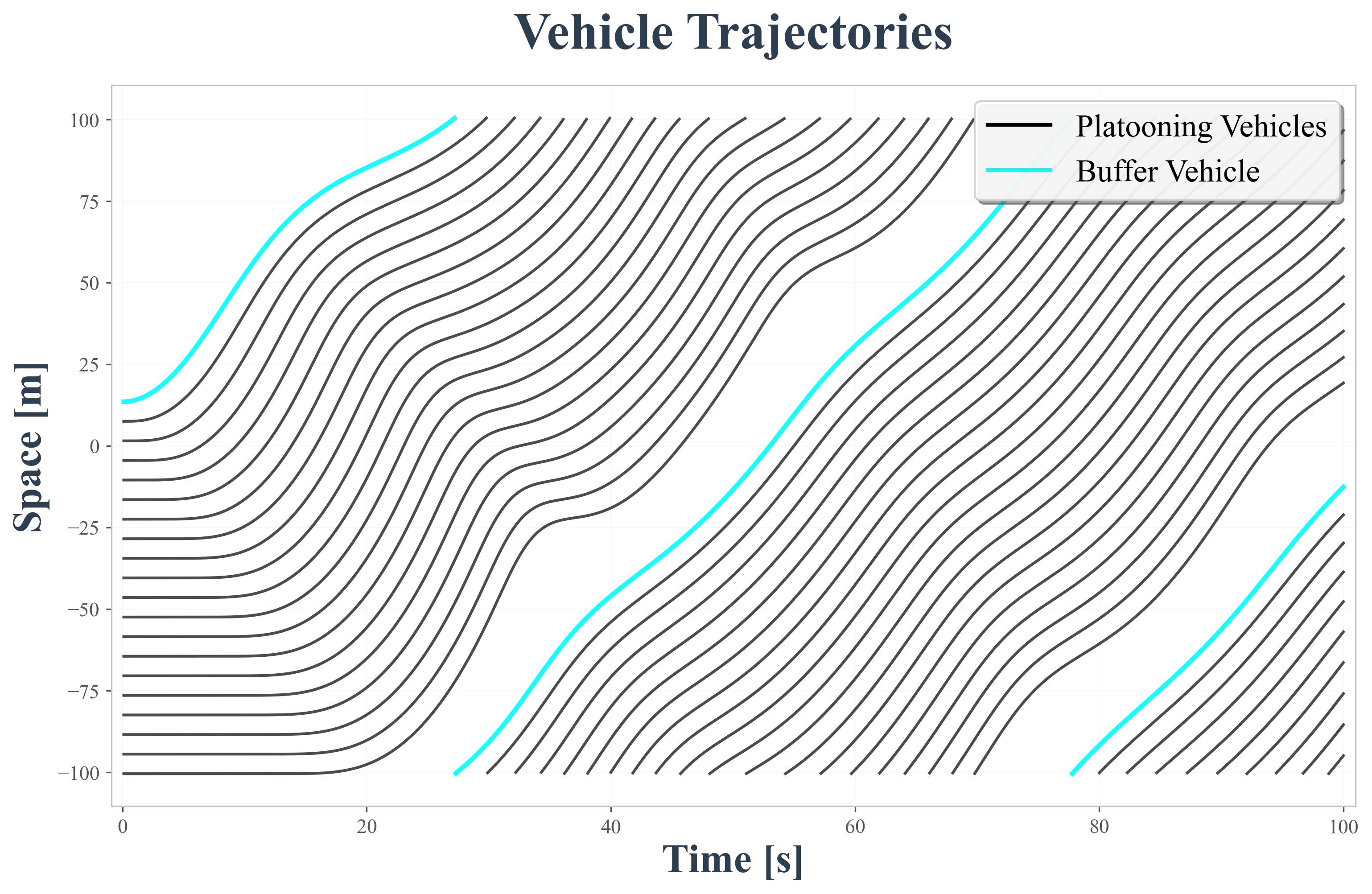}} \\ \vspace{1mm}
        \subfigure[]{\includegraphics[width=\textwidth]{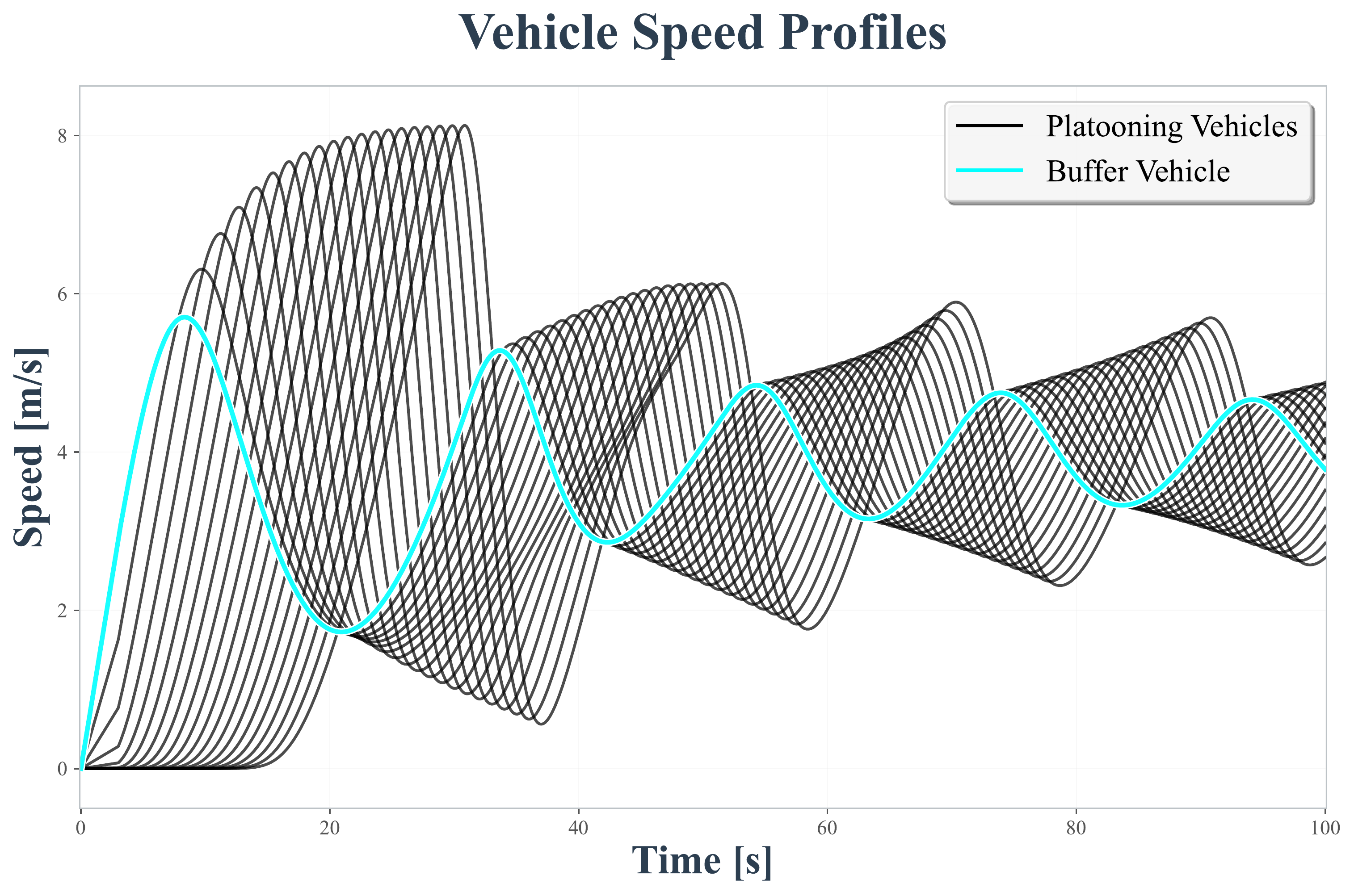}}
    \end{minipage}

    \caption{Simulation results using identical IDM parameter settings, without (left panels) and with (right panels) the buffering cooperative strategy. Both scenarios start from the same initial traffic jam condition. In (a), the trajectories of 20 non-cooperative vehicles are shown, with the corresponding speed profiles in (b). A stop-and-go wave emerges. The right panels (c, d) display the setup with the buffering cooperative strategy, where the flow tends to be laminar. The simulations can be computed online~\cite{SimBufferingStrategy}.}
    \label{fig_strategy}
\end{figure}

This paper challenges the foundational assumption of behavioral uniformity. Mathematically, a cooperative driving strategy is derived, and it is analytically proved that macroscopic throughput is globally maximized when a single designated vehicle assumes an enlarged headway, acting as a ``buffer'', while the remaining fleet cooperatively
forms a platoon. Through a rigorous optimization
formulation constrained by linear string stability conditions, we  establish that this non-uniform configuration is the global optimum. This analytical result is confirmed via numerical spectral analysis and maximum-stable-flow computations using the Intelligent Driver Model~(IDM) and the Full Velocity Difference Model~(FVDM), demonstrating that the strategy both expands the string-stable parameter space and achieves substantially higher macroscopic throughput than uniform baselines. This reframes the problem from cooperative controller design to cooperative \emph{role allocation} under stability constraints. The core mechanism and the SGW dissipation effect are illustrated in Figure~\ref{fig_strategy} as a proof-of-concept.

\section{Mathematical Formalization of Optimal Cooperative Buffering}
\label{sec:optimization}

Analysing traffic stability has a longer history, beginning with the first car-following models in the 1950s, with pioneering contributions from Reuschel~\cite{reuschel1950vehicle} and Pipes~\cite{pipes1953operational}. The question of traffic stability and SGWs inherent in delayed first-order and second-order models soon emerged and remains under investigation today (see, e.g., reviews~\cite{wilson2011car,lhe2019adaptive,korbmacher2025understanding}). 
In this section, the theoretical foundation of the cooperative buffering solution is established. Using a generalized linear car-following framework, it is demonstrated mathematically that a single-vehicle buffering strategy maximizes the mean vehicle speed while mathematically guaranteeing linear string stability. By doing so, it is proved that cooperative non-uniformity strictly outperforms any non-cooperative, uniform driving strategy.

\subsection{The Optimal Stability Control Problem}

Consider $N\ge 2$ vehicles on a one-dimensional road of length $L>0$ with periodic boundary conditions, i.e., $x_{N+1}=x_1+L$ and $v_{N+1}=v_1$.
The longitudinal dynamics are described by a linear heterogeneous car-following model:
\begin{equation}
\frac{d}{dt}v_n(t) = a_n\, g_n(t) + b_n\, v_n(t) + c_n\, v_{n+1}(t),
\label{eq:linear_model}
\end{equation}
where $v_n(t)=\frac{d}{dt}x_n(t)$ is the speed of vehicle $n$ and $g_n(t)=x_{n+1}(t)-x_n(t)-\ell_n$ is its gap to the leading vehicle.
Here $\ell_n\ge 0$ is the effective length of the vehicle, i.e., its physical length plus a minimum standstill distance. The parameters satisfy $a_n>0$, $b_n<0$, $c_n\ge 0$, and
$b_n+c_n<0$ for all $n\in\{1,\dots,N\}$.

A uniform equilibrium is a stationary state in which all vehicles travel at the same speed $v>0$ and keep constant gaps $g_n$, so that all accelerations vanish:
\begin{equation}
a_n g_n + (b_n+c_n)\,v = 0, \qquad \forall n\in\{1,\dots,N\}.
\label{eq:Equilibrium}
\end{equation}
This equilibrium is string stable in the longest wave length if and only if the following condition holds
\cite{ngoduy2015effect}:
\begin{equation}
\sum_{n=1}^N \frac{b_n^2-c_n^2}{2a_n^2} \ge \sum_{n=1}^N \frac{1}{a_n}.
\label{eq:general_stability}
\end{equation}
Note the condition simplify to the well established condition $b^2-c^2\ge2a$ for uniform parameters. 
We now parametrize the equilibrium through the time gap $T_n:=g_n/v$ of each vehicle.
This corresponds to the \emph{constant time gap} (CTG) policy, a common principle in many car-following models, whereby the equilibrium gap is proportional to the speed. 
The policy is historically motivated by a rule attributed to the \emph{California Motor Vehicle Code} Pipes~\cite{pipes1953operational,Chandler1958}. 
The same principle ensures safe following behavior and remains in use in adaptive cruise control (ACC) systems. 
Indeed, the ISO 15622 standard for ACC systems formalizes it, requiring a minimum selectable time gap
$\tau_{\min}\ge 0.8$~s, with at least one setting offered in the range $\tau\in[1.5,2.2]$~s for higher speeds~\cite{ISO15622:2018}.
Using \eqref{eq:Equilibrium}, the time gap of vehicle $n$ is
\begin{equation}
T_n=\frac{g_n}{v}=-\frac{b_n+c_n}{a_n}
\quad\Longleftrightarrow\quad
b_n=-a_nT_n-c_n, \qquad \forall n\in\{1,\dots,N\}.
\label{eq:CondPara}
\end{equation}
Hence, in our framework, the CTG policy is not an additional assumption but a direct consequence of the linear model \eqref{eq:linear_model}. 
The time gap is entirely determined by the model parameters $(a_n,b_n,c_n)$ and, conversely, can be prescribed by choosing them.
The standstill distance is accounted for by the effective length $\ell_n$ in the definition of $g_n$, so that the net gap vanishes with the speed. 
Allowing $T_n$ to differ across vehicles accounts for heterogeneous driver or ACC settings.
Substituting the expression \eqref{eq:CondPara} into \eqref{eq:general_stability}, and using
$b_n^2-c_n^2=a_n^2T_n^2+2a_nc_nT_n$, the stability condition becomes
\begin{equation}
\sum_{n=1}^N \Big(T_n^2 + 2\,T_n\,\frac{c_n}{a_n}\Big) \ge \sum_{n=1}^N \frac{2}{a_n}=: \alpha_N > 0.
\label{eq:stability_constraint2}
\end{equation}
The threshold $\alpha_N$ is positive since $a_n>0$.

At equilibrium, all the gaps satisfy $g_n=vT_n$.
On the closed road, the total gap is conserved:
$\sum_{n=1}^N g_n = L-\sum_{n=1}^N \ell_n =: \bar L$, where $\bar L>0$ is the total free space.
The equilibrium speed is therefore determined by the time gaps:
\begin{equation}
\sum_{n=1}^N vT_n=\bar L
\quad\iff\quad
v(T_1,\dots,T_N)=\frac{\bar L}{\sum_{n=1}^N T_n}.
\label{eq:mean_speed}
\end{equation}

Since $v$ is inversely proportional to $\sum_n T_n$, maximizing the flow $q=Nv/L$ (for fixed $N$ and $L$) is equivalent to minimizing the total time gap $\sum_n T_n$. 
The stability condition \eqref{eq:stability_constraint2}, however, imposes a competing requirement: the sum
of squared time gaps must exceed a model-dependent threshold. This tension between an $L_1$ objective and an $L_2$ constraint makes the problem non-trivial and gives it its characteristic geometry. 

For a general heterogeneous fleet ($c_n\ge 0$), let $\beta_n:=c_n/a_n\ge 0$ and recall that
$\alpha_N=\sum_{n=1}^N 2/a_n>0$. By \eqref{eq:stability_constraint2}, the optimal stability
control problem reads:

\smallskip
\textbf{Minimize:}
\begin{equation}
J(T_1,\dots,T_N)=\sum_{n=1}^N T_n
\label{eq:objective}
\end{equation}

\textbf{Subject to:}
\begin{align}
&\sum_{n=1}^N \big(T_n^2+2\beta_n T_n\big) \ge \alpha_N \quad \text{(stability constraint)}, \label{eq:constraint_1}\\[1mm]
&T_n \ge T_\text{min}, \quad \forall n\in\{1,\dots,N\}, \label{eq:constraint_2}
\end{align}
where $T_\text{min}>0$ is the minimum admissible time gap, assumed small enough that
\begin{equation}
C:=\alpha_N-\sum_{n=1}^N\big(T_\text{min}^2+2\beta_nT_\text{min}\big)
=\alpha_N-NT_\text{min}^2-2T_\text{min}\sum_{n=1}^N\beta_n>0.
\label{eq:constraint_3}
\end{equation}
Assumption \eqref{eq:constraint_3} states that the configuration $T_n=T_\text{min}$ for all $n$
violates the stability constraint. Without it, this configuration would be admissible and
trivially optimal. For $c_n=0$ we have $\beta_n=0$, the stability constraint reduces to
$\sum_n T_n^2\ge\alpha_N$, and \eqref{eq:constraint_3} reads $\alpha_N>NT_\text{min}^2$.
The feasible set is the region $\{T_n\ge T_\text{min}\ \forall n\}$ lying outside the open ball of center $-(\beta_1,\dots,\beta_N)$ and radius $\sqrt{\alpha_N+\sum_n\beta_n^2}$. It is therefore non-convex, which is the source of the
structure of the solution derived below.

\subsection{The Uniform (Non-Cooperative) Solution}

A natural benchmark, in line with classical homogeneous car-following models, is that all
drivers adopt the same time gap, $T_n=T$ for all $n$. Under this uniform strategy, the stability
constraint \eqref{eq:constraint_1} reads
\begin{equation}
T^2+2\bar\beta\,T-\frac{\alpha_N}{N}\ge 0,
\qquad \bar\beta:=\frac1N\sum_{n=1}^N\beta_n .
\end{equation}
The smallest admissible uniform time gap is the positive root
\begin{equation}
T_\text{non}^*=-\bar\beta+\sqrt{\bar\beta^{\,2}+\frac{\alpha_N}{N}},
\label{eq:T_non}
\end{equation}
which reduces to $T_\text{non}^*=\sqrt{\alpha_N/N}$ when $c_n=0$.
Since $h(T):=T^2+2\bar\beta T-\alpha_N/N$ is increasing for $T>0$ and vanishes at
$T_\text{non}^*$, we have $T_\text{non}^*>T_\text{min}$ if and only if $h(T_\text{min})<0$,
which is exactly \eqref{eq:constraint_3}. Hence the uniform solution is admissible, and by
\eqref{eq:mean_speed} its mean speed is
\begin{equation}
v_\text{non}^*=\frac{\bar L}{N\,T_\text{non}^*}
\qquad\Big(=\frac{\bar L}{\sqrt{N\alpha_N}}\ \text{ if } c_n=0\Big).
\label{eq:v_non}
\end{equation}
Any uniform time gap smaller than $T_\text{non}^*$ violates the stability constraint. Thus
$T_\text{non}^*$ is optimal among uniform strategies. Whether it is optimal among
\emph{all} admissible configurations is the question addressed next.

\subsection{The Cooperative Optimal Solution}

We now relax the uniformity assumption and allow non-uniform time-gap distributions. Write
\begin{equation}
T_n=T_\text{min}+u_n,\qquad u_n\ge 0,\qquad s:=\sum_{n=1}^N u_n .
\end{equation}
The objective becomes $J=NT_\text{min}+s$, so minimizing $J$ is equivalent to minimizing the
total excess gap $s$. Substituting into \eqref{eq:constraint_1} and using the definition
\eqref{eq:constraint_3} of $C$, the stability constraint becomes
\begin{equation}
\sum_{n=1}^N u_n^2+2\sum_{n=1}^N m_nu_n\ \ge\ C,
\qquad m_n:=T_\text{min}+\beta_n .
\label{eq:u_constraint}
\end{equation}
Let $m^*:=\max_n m_n=T_\text{min}+\beta_{\max}$ with $\beta_{\max}:=\max_n\beta_n$. Since $u_n\ge0$,
\begin{equation}
\sum_{n=1}^N m_nu_n\le m^*s,
\qquad
\sum_{n=1}^N u_n^2\le\Big(\sum_{n=1}^N u_n\Big)^2=s^2 .
\label{eq:inequality}
\end{equation}
The second inequality holds because the cross terms $u_nu_m$ are non-negative. Combining
\eqref{eq:u_constraint} and \eqref{eq:inequality}, every admissible configuration satisfies
\begin{equation}
s^2+2m^*s\ \ge\ C .
\end{equation}
The function $s\mapsto s^2+2m^*s$ is increasing on $s\ge0$ and vanishes at $s=0<C$. Hence
\begin{equation}
s\ \ge\ s^*:=-m^*+\sqrt{(m^*)^2+C}\ >0 .
\label{eq:s_star}
\end{equation}

The bound \eqref{eq:s_star} is attained if and only if both inequalities in
\eqref{eq:inequality} are equalities. The second one, $\sum_n u_n^2=(\sum_n u_n)^2$, holds for
non-negative variables if and only if at most one $u_n$ is non-zero
\cite{korbmacher2025emergent}. The first one holds if and only if every $n$ with $u_n>0$
satisfies $m_n=m^*$, i.e., $\beta_n=\beta_{\max}$. This gives the following result.

\begin{proposition}[Optimal cooperative time gaps]\label{prop:coop}
Under \eqref{eq:constraint_3}, the optimal stability control problem has a minimizer. The
minimizers are exactly the configurations
\begin{align}
T_n^* &= T_\text{min}, \qquad n\neq n_0, \label{eq:opt_platoon}\\
T_{n_0}^* &= T_\text{min}+s^*
=-\beta_{\max}+\sqrt{\beta_{\max}^2+\alpha_N-\sum_{n\neq n_0}\big(T_\text{min}^2+2\beta_nT_\text{min}\big)},
\label{eq:opt_buffer}
\end{align}
where the ``buffer vehicle'' $n_0$ belongs to $\arg\max_n\, c_n/a_n$. For $c_n=0$, every
index is admissible as buffer vehicle, and \eqref{eq:opt_buffer} reads
$T_{n_0}^*=\sqrt{\alpha_N-(N-1)T_\text{min}^2}$.
\end{proposition}

\begin{proof}
By the chain of inequalities above, every admissible configuration satisfies $s\ge s^*$. The
configuration \eqref{eq:opt_platoon}--\eqref{eq:opt_buffer} is admissible (it saturates
\eqref{eq:u_constraint}) and has $s=s^*$, so it is a global minimizer. Conversely, if $s=s^*$,
all inequalities in the chain are equalities, which forces the structure described above.
The expression in \eqref{eq:opt_buffer} follows from
$(T_\text{min}+\beta_{\max})^2+C=\beta_{\max}^2+\alpha_N-\sum_{n\ne n_0}(T_\text{min}^2+2\beta_nT_\text{min})$.
\end{proof}

The maximum achievable stable throughput thus requires an extreme behavioral non-uniformity:
$N-1$ vehicles form a dense, high-capacity platoon operating at the minimum time gap
$T_\text{min}$, while a single cooperative vehicle adopts a large time gap and acts as a
buffer that stabilizes the whole fleet. When the vehicles are heterogeneous, the buffer role
is assigned to the vehicle with the largest ratio $c_n/a_n$, i.e., the one for which an
increase of the time gap yields the largest gain in stability. The corresponding mean speed is
\begin{equation}
v_\text{coop}^*=\frac{\bar L}{(N-1)T_\text{min}+T_{n_0}^*}.
\end{equation}

\begin{corollary}[Cooperation strictly outperforms uniformity]\label{cor:comparison}
Under \eqref{eq:constraint_3}, $v_\text{coop}^*>v_\text{non}^*$.
\end{corollary}
\begin{proof}
The uniform solution $T_n\equiv T_\text{non}^*$ is admissible (Section on the uniform
solution), and therefore $\sum_n T_n^*\le N\,T_\text{non}^*$ by optimality of the cooperative
solution. Since $N\ge2$ and $T_\text{non}^*>T_\text{min}$, the uniform configuration has
at least two components strictly above $T_\text{min}$, so by Proposition~\ref{prop:coop} it is
not a minimizer. The inequality is therefore strict, and the result follows from
\eqref{eq:mean_speed}.
\end{proof}



\begin{remark}[Mean--variance reformulation]
\label{rem:mean_variance}
Let $\bar T,\bar\beta$ denote the means of $T_n,\beta_n:=c_n/a_n$ over the $N$ vehicles,
$\sigma_T^2,\sigma_\beta^2$ their empirical variances, $\mathrm{Cov}(T,\beta)$ their
covariance, and $a_H:=N/\sum_n a_n^{-1}$ the harmonic mean of the $a_n$. Using
$\frac1N\sum_nT_n^2=\bar T^2+\sigma_T^2$ and $\frac1N\sum_n\beta_nT_n=\bar\beta\bar
T+\mathrm{Cov}(T,\beta)$, the stability condition \eqref{eq:stability_constraint2} is
equivalent to
\begin{equation}
\bar T^2+\sigma_T^2+2\bar\beta\,\bar T+2\,\mathrm{Cov}(T,\beta)\ \ge\ \frac{2}{a_H}
=\frac{\alpha_N}{N},  \qquad \beta_n:=\frac{c_n}{a_n}.
\label{eq:mean_var_general}
\end{equation}
For $c_n=0$ (i.e., $\beta_n=0$) this reduces to $\bar T^2+\sigma_T^2\ge\alpha_N/N$.
\end{remark}

\begin{remark}[What stabilizes the fleet]
By \eqref{eq:mean_speed}, throughput depends only on $\bar T$, so any stabilizing effect that
does not raise $\bar T$ is obtained for free. 
Equation~\ref{eq:mean_var_general} shows that the three parameters $(T_n,a_n,c_n)$ play markedly different roles in this respect.
\begin{itemize}
    \item \emph{Time gaps $T_n$.~} The variance $\sigma_T^2$ enters \eqref{eq:mean_var_general} directly
and always helps improving the stability.  
Dispersion of the $T_n$ is \emph{unconditionally} stabilizing, at no cost to the mean. 
This is the mechanism behind Proposition~\ref{prop:coop}: the optimal stability control problem is exactly a maximization of $\sigma_T^2$ at fixed $\bar T$, subject to $T_n\ge T_\text{min}$.

\item \emph{Coupling coefficients $c_n$.~} The dependence on $c_n/a_n$ in \eqref{eq:mean_var_general} is exactly linear, so, unlike $T_n$, the variance $\sigma_{c/a}^2$ plays no role whatsoever. 
Only the mean $\overline{c/a}$ (which lowers the required $\bar T$, even for a uniform fleet) and the covariance $\mathrm{Cov}(T,c/a)$ matter. 
Since $|\mathrm{Cov}(T,c/a)|\le\sigma_T\sigma_{c/a}$, heterogeneity in $c_n$ is stabilizing only insofar as large time gaps are assigned to vehicles with large $c_n/a_n$.
This is consistent with the choice of buffer vehicle $n_0\in\arg\max_n c_n/a_n$ in Proposition~\ref{prop:coop}.

\item \emph{Reactivity coefficients $a_n$.~} The parameter $a_n$ enters \eqref{eq:mean_var_general} twice: through the harmonic mean $a_H$ on the right-hand side, and through $\beta_n=c_n/a_n$ on the left-hand side. 
Through $a_H$ alone (i.e., $c_n\equiv0$), heterogeneity of the $a_n$ is destabilizing at fixed mean $\bar a$, since $a_H\le\bar a$. 
This is opposite to the effect of $T_n$. 
Once $c_n\neq0$, this conclusion no longer holds in general: it is reversed for vehicles with $c_nT_n>1$. 
More broadly the dominant effect is the correlation between $a_n$ and $c_nT_n$ across the fleet (small $a_n$ assigned to vehicles with large $c_nT_n$ is stabilizing) rather than the variance of $a_n$ in isolation.
\end{itemize}

In summary, only the heterogeneity of the time gaps $T_n$ stabilizes systematically. 
The effect of heterogeneity in $a_n$ and $c_n$ is more intricate and is governed chiefly by their correlation with $T_n$, rather than by their own variance (see Table~\ref{tab:stabilizing_mechanisms}).

\begin{table}[htbp]
\centering
\caption{Role of the mean and the heterogeneity of each parameter in the stability condition
\eqref{eq:stability_constraint2}, at fixed values of the other two parameters.\medskip}
\label{tab:stabilizing_mechanisms}
\begin{tabular}{lp{4.2cm}p{7.7cm}}
\toprule
Parameter & Effect of the mean & Effect of heterogeneity (fixed mean) \\
\midrule
$T_n$ & Raising $\bar T$ stabilizes, but reduces throughput &
\emph{Always} stabilizing, at no cost to throughput; enters exactly through $\sigma_T^2$ \\[2mm]
$c_n$ & Raising $\bar c$ stabilizes, at no cost to throughput &
No effect in isolation ($\sigma_c^2$ does not appear); stabilizing only through
$\mathrm{Cov}(T_n,c_n/a_n)>0$, i.e., large time gaps assigned to large $c_n/a_n$ \\[2mm]
$a_n$ & Raising $\bar a$ stabilizes if $c_nT_n<1$, at no cost to throughput &
Destabilizing if $a_n$ and $c_nT_n$ are positively correlated (in particular 
if $c_nT_n<1$ uniformly); reversed if $c_nT_n>1$ \\[0mm]
\bottomrule
\end{tabular}
\end{table}
\end{remark}

\section{Microscopic Models}
While the analytical proof in Section~\ref{sec:optimization} guarantees the global optimality of a non-uniform time-gap distribution under linearized assumptions, real-world traffic is governed by complex, nonlinear dynamics. 
In addition, the string stability condition \eqref{eq:general_stability} only holds for the longest wavelength. However, shorter wavelength modes can become unstable for non-uniform models \cite{treiber2012traffic}.
To demonstrate that the cooperative buffering strategy effectively dissipates finite-amplitude perturbations and expands the basin of attraction for laminar flow, the system is evaluated using two fundamentally distinct second-order car-following models. By demonstrating the strategy's efficacy across dynamically divergent models, we verify that the theoretical benefits of non-uniformity are robust, fundamental to traffic physics, and strictly model-agnostic.

\subsection{Intelligent Driver Model}
The IDM, proposed by Treiber et al.~\cite{treiber2000congested}, is a continuous-time model that smoothly combines a free-road acceleration strategy with a collision-avoidance braking strategy. The continuous acceleration of the $n$-th vehicle is given by:
\begin{equation}
    \frac{dv_n}{dt} = a \left[ 1 - \left(\frac{v_n}{v_0}\right)^4 - \left(\frac{s_n^*(v_n, v_{n+1})}{g_n}\right)^2 \right]
\end{equation}
where the desired dynamic gap $s^*$ is calculated as:
\begin{equation}
     s_n^*(v_n, v_{n+1}) = s_0 + v_n T_n + v_n\frac{v_n-v_{n+1}}{2\sqrt{ab}}
     \label{sstar}
\end{equation}
Here, $v_0>0$ is the desired free-flow velocity, 
$s_0\ge0$ is the minimum stopping distance, $T_n>0$ is the desired time gap specific to each vehicle, $a>0$ is the maximum acceleration, and $b>0$ is the comfortable deceleration parameter governing passenger ride comfort. The IDM is highly realistic but notoriously prone to string instability for small values of $a$ and high values of $b$.

\subsection{Full Velocity Difference Model (FVDM)}
The FVDM, developed by Jiang et al.\ \cite{jiang2001full}, extends the classical Optimal Velocity Model~\cite{bando1995dynamical} by explicitly incorporating 
the relative velocity with the predecessor. The dynamic equation is:
\begin{equation}
    \frac{dv_n}{dt} = \frac{1}{\tau_1} \big(V_n(g_n) - v_n\big) + \frac{1}{\tau_2} \big(v_{n+1}-v_n\big)
\end{equation}
where $V_n(g_n)=\min\big\{v_0,\max\{0,g_n/T_n\}\big\}$ is the piece-wise linear optimal velocity function depending on the inter-vehicle gap $g_n$, $\tau_1>0$ represents the driver's relaxation time (sensitivity to the optimal speed), and $\tau_2>0$ represents the sensitivity to the relative velocity. 
Note that the OV function is vehicle-specific, since the time gap parameter ($T_n$) differs between vehicles.
The FVDM allows us to validate the optimal cooperative buffering strategy on a model whose instability dynamics are governed primarily by relaxation mechanisms rather than direct kinematic constraints as in IDM.

\subsection{Implementation of the Cooperative Buffering Strategy}

  \begin{figure}[tbp]
    \centering

    \subfigure[]{
        \includegraphics[width=0.48\textwidth]{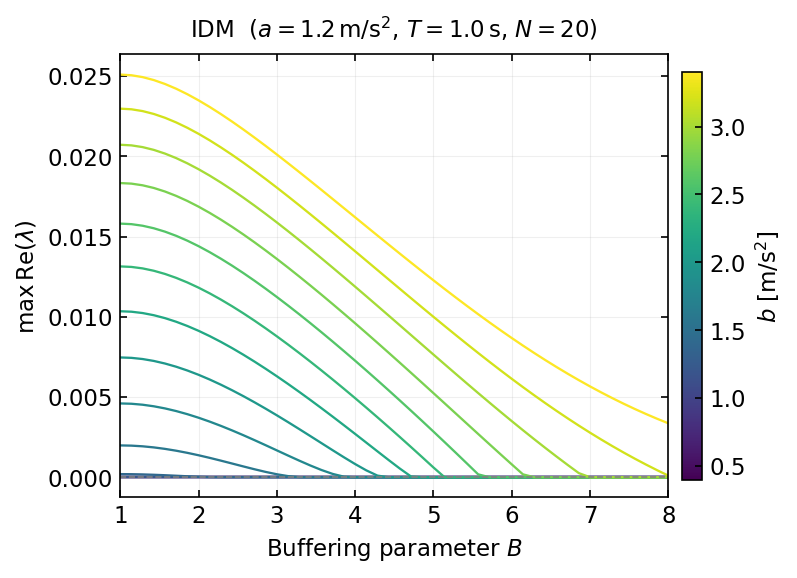}
        \label{RE_a} 
    }
    \hfill 
    \subfigure[]{
        \includegraphics[width=0.48\textwidth]{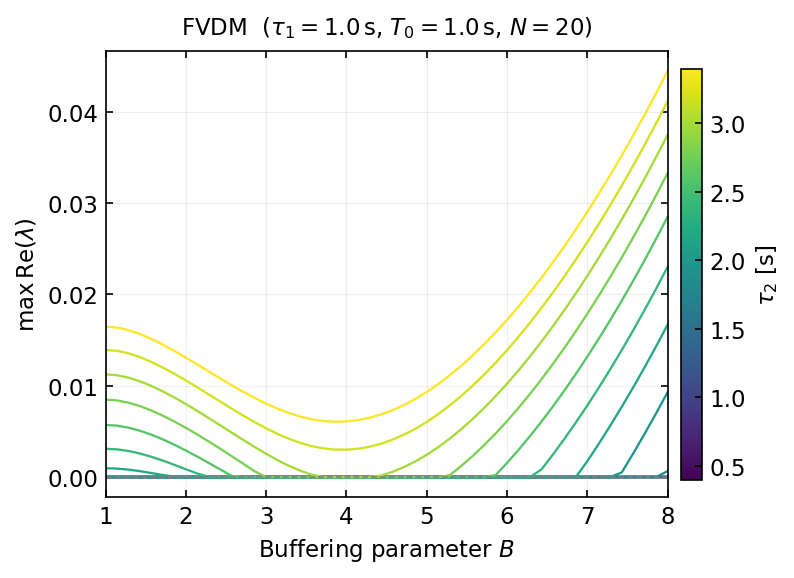}
        \label{RE_b} 
    }

    \caption{Maximum real part of the dominant eigenvalue, $\max Re(\lambda)$, as a function of the buffering parameter B for a platoon of $N = 20$ vehicles. Each curve corresponds to a different value of the respective sensitivity parameter (color-coded). The dashed line at max $Re(\lambda) = 0$ marks the string-stability boundary. (a) IDM: increasing $b$ raises the instability growth rate, while increasing B monotonically reduces it, eventually restoring string stability. (b) FVDM: a non-monotonic dependence on $B$ is observed; moderate buffering stabilizes the platoon, but excessive buffering
  re-introduces instability, particularly for large values of $\tau_2$.}
    \label{fig:spectral} 
\end{figure}

A periodic ring-road topology of length $L = 201$ meters populated by $N = 20$ vehicles similar to the setup of Sugiyama et al. 2008~\cite{sugiyama2008traffic} is considered. This spatial configuration establishes a capacity-constrained, dense traffic regime inherently susceptible to the spontaneous emergence of SGWs. 
To operationalize the theoretically optimal cooperative buffering strategy within these nonlinear frameworks, a dimensionless buffering scaling parameter $B \ge 1$ is introduced. For the single designated buffer vehicle, the nominal time gap $T$ is amplified such that $T_{\text{buffer}} = B \cdot T$. Crucially, to rigorously isolate the stabilizing effect of spatial non-uniformity from trivial density reductions, we enforce a strict mean-preserving transformation. The remaining $N-1$ platooning vehicles cooperatively compress their respective time gaps to
\begin{equation}
T_{\text{platoon}} = \frac{N - B}{N - 1} T.
\label{platoon}
\end{equation}
Consequently, the system-wide average time gap is strictly maintained at $T$. The traditional uniform traffic paradigm corresponds to the homogeneous baseline state of $B = 1$.

\section{Numerical Spectral Analysis}

The efficacy of the cooperative buffering strategy is evaluated across two critical dimensions. First, the suppression of string instabilities via numerical spectral analysis, and second, the maximization of string stable traffic flow.

\subsection{Spectral Analysis}
\label{sec:spectral}

The string stability of the cooperatively buffered traffic flow is first assessed by computing the maximum real part of the eigenvalues, $\max\operatorname{Re}(\lambda)$, derived from the Jacobian matrix of the linearized non-uniform system evaluated at the uniform equilibrium state.
In contrast to condition~(\ref{eq:general_stability}), which governs the longest-wavelength mode, the numerical spectral analysis evaluates stability across the full spectrum of wavenumbers $\theta\in[0,\pi]$, thereby capturing potential short-wavelength instabilities that arise exclusively in heterogeneous configurations~\cite{treiber2012traffic}. The traffic flow is linearly string-stable if and only if
$\max\operatorname{Re}(\lambda)\leq 0$ across all modes.

Figure~\ref{RE_a} presents the stability spectrum for the IDM across varying
comfortable deceleration parameters~$b$, setting the desired velocity to
$v_0=20$\,m/s, the nominal time gap to $T_0=1.0$\,s, and maximum acceleration to $a=1.2$\,m/s$^2$. Under uniform, non-cooperative conditions ($B=1$), the system is acutely unstable ($\max\operatorname{Re}(\lambda)>0$)
across a broad spectrum of comfort preferences. However, as the cooperative
buffering multiplier~$B$ increases, the dominant eigenvalue's real part monotonically decreases. For any given deceleration threshold, allocating a sufficiently large buffer to a single vehicle successfully drives the system into the stable left-half plane. 

Figure~\ref{RE_b} illustrates the corresponding eigenvalue spectrum for the FVDM
across varying relative velocity sensitivities~$\tau_2$. Unlike the
monotonic stabilization observed in the IDM, the FVDM exhibits a profound,
non-monotonic stability landscape. Moderate cooperative buffering
($3\leq B\leq 5$) successfully dampens the dominant long-wavelength
instabilities, minimizing $\max\operatorname{Re}(\lambda)$. However, excessive buffering ($B>5$) triggers a resurgence of instability. An exceptionally large buffer mandates that the remaining $N-1$ platooning vehicles operate at severely compressed headways, which violates the microscopic relaxation limits of the FVDM and excites high-frequency, short-wavelength modes. Nevertheless, the optimally tuned cooperative system ($B\approx 4$) vastly outperforms the uniform baseline, confirming that the stabilising mechanism of cooperative buffering is robust even in models whose instability dynamics are relaxation-driven.

\subsection{Effect on Traffic Flow}

Having established that cooperative buffering expands the string stable parameter space, this section quantifies the direct impact on traffic efficiency. For each model and parameter configuration, the
\emph{maximum stable flow}~$Q$ (veh/h) is determined. This describes the highest throughput at which the system remains linearly string-stable across all wavenumbers. This quantity represents the operational capacity frontier. Any further increase in flow would necessarily violate the stability constraint and permit SGW formation.

The computation of~$Q$ reflects the distinct equilibrium structures of each model. In the IDM, the nominal time gap~$T_0$ is held constant and different operating points on the fundamental diagram are explored by
varying the equilibrium free spacing~$s_e$. For each value of~$B$, the highest equilibrium velocity is identified, and hence the highest flow $Q=
v_e(s_e)/(s_e+\ell)$, where~$s_e$ denotes the IDM equilibrium
spacing equation~(\ref{sstar})\, for which all eigenvalues of the linearized non-uniform system have non-positive real parts, while $\ell=l+s_0$. In the FVDM, vehicles are assumed to operate at the free-flow speed~$v_0$, and the capacity frontier is obtained by minimizing the nominal time gap~$T_0$ subject to the constraint that the heterogeneous system remains string-stable, yielding
$Q=v_0/(\ell+v_0\,T^*)$. 
In both cases, the precise stability boundaries are located via bisection.

\begin{figure}[tbp]
    \centering

    \subfigure[]{
        \includegraphics[width=0.48\textwidth]{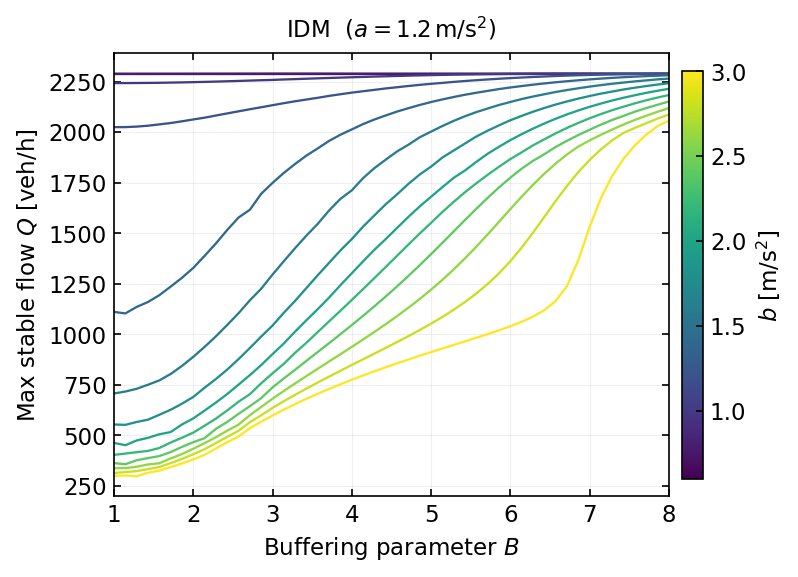}
        \label{fig:ts_a} 
    }
    \hfill 
    \subfigure[]{
        \includegraphics[width=0.48\textwidth]{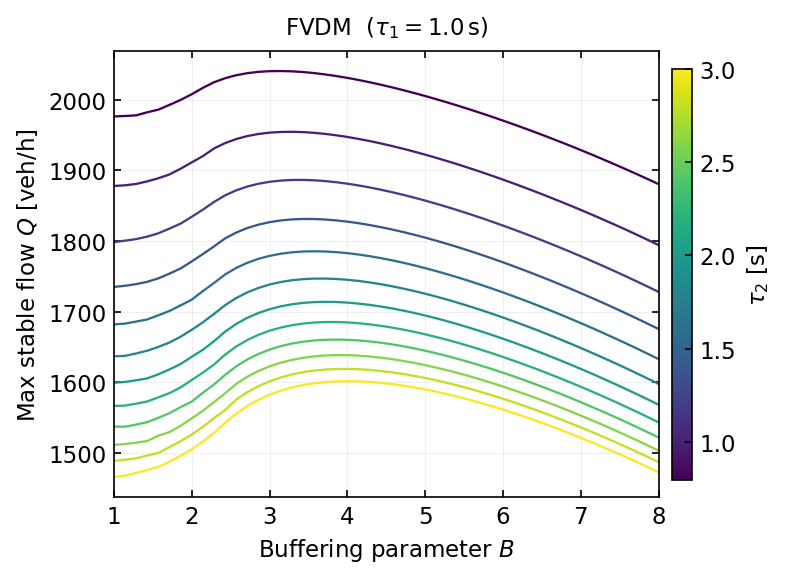}
        \label{fig:ts_b} 
    }

    \caption{ Maximum string-stable traffic flow $Q$ [veh/h] as a function of the cooperative buffering parameter $B$. The curves indicate the highest theoretical throughput achievable while preserving linear string stability ($\max \text{Re}(\lambda) \le 0$). (a)~IDM: Evaluated across varying comfortable deceleration parameters $b$. For highly unstable configurations (high $b$), the uniform baseline ($B=1$) suffers severe capacity drops. Increasing $B$ effectively dissipates perturbations and restores near-theoretical maximum throughput. (b)~FVDM: Evaluated across varying relative velocity sensitivities $\tau_2$. The macroscopic throughput exhibits a non-monotonic dependence on $B$, achieving an absolute global maximum at an optimally tuned moderate buffer ($B \approx 3$ to $4$). In both models, optimal cooperative buffering strictly outperforms the uniform baseline.}
    \label{fig:stable_flow} 
\end{figure}

Figure~\ref{fig:ts_a} presents the maximum stable flow for the IDM as a function of the buffering parameter~$B$, with curves color-coded by the comfortable deceleration parameter~$b$. At the uniform baseline ($B=1$), the achievable stable throughput is severely constrained. This applies especially for large
values of~$b$, where the IDM's heightened susceptibility to string instability forces the system to operate at conservative, low-density equilibria. As~$B$ increases, the maximum stable flow rises monotonically
and substantially. For high comfort parameters ($b\geq 2.5$\,m/s$^2$), the cooperative strategy unlocks high throughput gains, demonstrating
that the single buffer vehicle absorbs the system's entire stability burden and frees the remaining $N-1$ vehicles to operate at higher-density, higher-flow equilibria that would be categorically unstable under uniform
conditions. Crucially, even for moderate deceleration preferences, the cooperative strategy consistently outperforms the uniform baseline at every
value of $B>1$.

Figure~\ref{fig:ts_b} displays the analogous analysis for FVDM, with curves parametrized by the relative velocity sensitivity~$\tau_2$. Consistent with
the non-monotonic spectral behavior identified in Section~\ref{sec:spectral}, the
maximum stable flow initially increases with~$B$ but saturates and, for
large~$\tau_2$, exhibits a slight decline at excessive buffering levels.
This saturation reflects the FVDM's relaxation dynamics. As the buffer gap grows, the compressed platoon headways approach the model's minimum
relaxation threshold, exciting short-wavelength instabilities that cap further capacity gains. The throughput improvement is most pronounced for large~$\tau_2$, where the uniform baseline is most severely flow-constrained by instability. For moderate buffering ($B\approx 3$--$5$), the cooperative
strategy yields a significant increase in throughput.

\section{Conclusion and Future Work}

This paper formalizes traffic string stability as a constrained optimization problem and proves analytically that macroscopic throughput is globally maximized by a non-uniform cooperative strategy.

This analytical result is verified via numerical spectral analysis and maximum-stable-flow computations using two fundamentally distinct car-following models. For the IDM, cooperative buffering monotonically expands the string-stable parameter space with increasing~$B$, resolving the classical trade-off between comfortable deceleration and string stability and unlocking throughput gains exceeding~40\% for realistic parameter settings. For the FVDM, whose dynamics are relaxation-driven, moderate buffering ($B\approx 3$--$5$) yields the largest capacity improvements, while excessive buffering triggers short-wavelength instabilities. By preserving the mean time gap through cooperative compression~(see eq.~(\ref{platoon})), the strategy also inherently outperforms non-cooperative JAD, which inflates average headways at the cost of reduced throughput.

While Equation~\eqref{eq:opt_buffer} proves that a single buffer is the theoretical global optimum, practical highway deployments introduce constraints not captured by the periodic ring-road topology. Excessively large buffer headways invite cut-in maneuvers from adjacent lanes, and long continuous platoons may excite short-wavelength instabilities, as observed in the FVDM. Future work will therefore extend this framework to a constrained multi-buffer optimization problem, deriving the optimal spatial distribution of multiple buffer vehicles under bounds on maximum platoon length and buffer gap size, bridging the analytical optimum established here with the operational realities of multi-lane highway traffic.
\begin{acknowledgements}
  The authors gratefully acknowledge Claudia Totzeck for her valuable contributions in deriving the optimality conditions, as well as support from the Deutsche Forschungsgemeinschaft (DFG, Germany) under grant number 546728715.

\end{acknowledgements}


\section*{Data and code availability}
An interactive demonstration of cooperative buffering is cited in Ref.~\cite{SimBufferingStrategy}.


\bibliographystyle{abbrv} 
\bibliography{biblio} 

@article{ciuffo2021requiem,
  doi = {10.1016/j.trc.2021.103305},
  title={Requiem on the positive effects of commercial adaptive cruise control on motorway traffic and recommendations for future automated driving systems},
  author={Ciuffo, Biagio and Mattas, Konstantinos and Makridis, Michail and Albano, Giovanni and Anesiadou, Aikaterini and He, Yinglong and Josvai, Szil{\'a}rd and Komnos, Dimitris and Pataki, Marton and Vass, Sandor and others},
  journal={Transportation Research Part C: Emerging Technologies},
  volume={130},
  pages={103305},
  year={2021},
  publisher={Elsevier}
}

@article{bando1995dynamical,
  doi = {10.1103/physreve.51.1035},
  title={Dynamical model of traffic congestion and numerical simulation},
  author={Bando, Masako and Hasebe, Katsuya and Nakayama, Akihiro and Shibata, Akihiro and Sugiyama, Yuki},
  journal={Physical review E},
  volume={51},
  number={2},
  pages={1035--1042},
  year={1995},
  publisher={APS}
}

@article{jiang2001full,
  doi = {10.1103/physreve.64.017101},
  title={Full velocity difference model for a car-following theory},
  author={Jiang, Rui and Wu, Qingsong and Zhu, Zuojin},
  journal={Physical Review E},
  volume={64},
  number={1},
  pages={017101},
  year={2001},
  publisher={APS}
}

@article{ngoduy2015effect,
  doi = {10.1080/21680566.2014.960503},
  title={Effect of the car-following combinations on the instability of heterogeneous traffic flow},
  author={Ngoduy, Dong},
  journal={Transportmetrica B: transport dynamics},
  volume={3},
  number={1},
  pages={44--58},
  year={2015},
  publisher={Taylor \& Francis}
}

@article{li2014stop,
  doi = {10.1016/j.trb.2014.09.014},
  title={Stop-and-go traffic analysis: Theoretical properties, environmental impacts and oscillation mitigation},
  author={Li, Xiaopeng and Cui, Jianxun and An, Shi and Parsafard, Mohsen},
  journal={Transportation Research Part B: Methodological},
  volume={70},
  pages={319--339},
  year={2014},
  publisher={Elsevier}
}

@inproceedings{korbmacher2025understanding,
  doi = {10.1109/cdc57313.2025.11312050},
  title={Understanding collective stability of {ACC} systems: {}From theory to real-world observations},
  author={Korbmacher, Raphael and Khound, Parthib and Tordeux, Antoine},
  booktitle={2025 IEEE 64th Conference on Decision and Control (CDC)},
  pages={8371--8378},
  year={2025},
  organization={IEEE}
}

@article{sugiyama2008traffic,
  doi = {10.1088/1367-2630/10/3/033001},
  title={Traffic jams without bottlenecks—experimental evidence for the physical mechanism of the formation of a jam},
  author={Sugiyama, Yuki and Fukui, Minoru and Kikuchi, Macoto and Hasebe, Katsuya and Nakayama, Akihiro and Nishinari, Katsuhiro and Tadaki, Shin-ichi and Yukawa, Satoshi},
  journal={New Journal of Physics},
  volume={10},
  number={3},
  pages={033001},
  year={2008},
  publisher={IOP Publishing}
}

@article{tadaki2013phase,
  doi = {10.1088/1367-2630/15/10/103034},
  title={Phase transition in traffic jam experiment on a circuit},
  author={Tadaki, Shin-ichi and Kikuchi, Macoto and Fukui, Minoru and Nakayama, Akihiro and Nishinari, Katsuhiro and Shibata, Akihiro and Sugiyama, Yuki and Yosida, Taturu and Yukawa, Satoshi},
  journal={New Journal of Physics},
  volume={15},
  number={10},
  pages={103034},
  year={2013},
  publisher={IOP publishing}
}

@article{jiang2014traffic,
  doi = {10.1371/journal.pone.0094351},
  title={Traffic experiment reveals the nature of car-following},
  author={Jiang, Rui and Hu, Mao-Bin and Zhang, HM and Gao, Zi-You and Jia, Bin and Wu, Qing-Song and Wang, Bing and Yang, Ming},
  journal={PloS one},
  volume={9},
  number={4},
  pages={e94351},
  year={2014},
  publisher={Public Library of Science San Francisco, USA}
}

@article{stern2018dissipation,
  doi = {10.1016/j.trc.2018.02.005},
  title={Dissipation of stop-and-go waves via control of autonomous vehicles: Field experiments},
  author={Stern, Raphael E and Cui, Shumo and Delle Monache, Maria Laura and Bhadani, Rahul and Bunting, Matt and Churchill, Miles and Hamilton, Nathaniel and Pohlmann, Hannah and Wu, Fangyu and Piccoli, Benedetto and others},
  journal={Transportation Research Part C: Emerging Technologies},
  volume={89},
  pages={205--221},
  year={2018},
  publisher={Elsevier}
}

@article{nishi2013theory,
  doi = {10.1016/j.trb.2013.02.003},
  title={Theory of jam-absorption driving},
  author={Nishi, Ryosuke and Tomoeda, Akiyasu and Shimura, Kenichiro and Nishinari, Katsuhiro},
  journal={Transportation Research Part B: Methodological},
  volume={50},
  pages={116--129},
  year={2013},
  publisher={Elsevier}
}

@article{hegyi2005optimal,
  doi = {10.1109/TITS.2004.842408},
  title={Optimal coordination of variable speed limits to suppress shock waves},
  author={Hegyi, Andreas and De Schutter, Bart and Hellendoorn, Johannes},
  journal={IEEE Transactions on Intelligent Transportation Systems},
  volume={6},
  number={1},
  pages={102--112},
  year={2005},
  publisher={IEEE}
}

@inproceedings{hegyi2010dynamic,
  doi = {10.1109/itsc.2010.5624974},
  title={Dynamic speed limit control to resolve shock waves on freeways-Field test results of the SPECIALIST algorithm},
  author={Hegyi, Andreas and Hoogendoorn, Serge P},
  booktitle={13th International IEEE Conference on Intelligent Transportation Systems},
  pages={519--524},
  year={2010},
  organization={IEEE}
}

@article{reuschel1950vehicle,
  title={Vehicle movements in a platoon},
  author={Reuschel, A},
  journal={Oesterreichisches Ingenieur-Archir},
  volume={4},
  pages={193--215},
  year={1950}
}

@article{pipes1953operational,
  doi = {10.1063/1.1721265},
  title={An operational analysis of traffic dynamics},
  author={Pipes, Louis A},
  journal={Journal of Applied Physics},
  volume={24},
  number={3},
  pages={274--281},
  year={1953},
  publisher={American Institute of Physics}
}

@ARTICLE{he2025review,
  doi = {10.1109/tits.2026.3658644},
  author={He, Zhengbing and Laval, Jorge A. and Han, Yu and Hegyi, Andreas and Nishi, Ryosuke and Wu, Cathy},
  journal={IEEE Transactions on Intelligent Transportation Systems}, 
  title={A Review of Stop-and-Go Traffic Wave Suppression Strategies: Variable Speed Limit Versus Jam-Absorption Driving}, 
  year={2026},
  volume = {27},
  number = {5},
  pages = {4986--5000}}

@article{wang2021jam,
  doi = {10.1109/mits.2021.3050889},
  title={Jam-absorption driving strategy for improving safety near oscillations in a connected vehicle environment considering consequential jams},
  author={Wang, Shunchao and Li, Zhibin and Cao, Zehong and Jolfaei, Alireza and Cao, Qi},
  journal={IEEE Intelligent Transportation Systems Magazine},
  volume={14},
  number={2},
  pages={41--52},
  year = {2022},
  publisher={IEEE}
}

@article{gunter2020commercially,
  doi = {10.1109/tits.2020.3000682},
  title={Are commercially implemented adaptive cruise control systems string stable?},
  author={Gunter, George and Gloudemans, Derek and Stern, Raphael E and McQuade, Sean and Bhadani, Rahul and Bunting, Matt and Delle Monache, Maria Laura and Lysecky, Roman and Seibold, Benjamin and Sprinkle, Jonathan and others},
  journal={IEEE Transactions on Intelligent Transportation Systems},
  volume={22},
  number={11},
  pages={6992--7003},
  year = {2021},
  publisher={IEEE}
}

@article{treiber2001microsimulations,
title = {Microsimulations of Freeway Traffic Including Control Measures},
author = {M. Treiber and D. Helbing},
pages = {478},
volume = {49},
number = {11},
journal = {at - Automatisierungstechnik},
doi = {10.1524/auto.2001.49.11.478},
year = {2001},
}

@article{davis2004effect,
  doi = {10.1103/physreve.69.066110},
  title={Effect of adaptive cruise control systems on traffic flow},
  author={Davis, LC},
  journal={Physical Review E},
  volume={69},
  number={6},
  pages={066110},
  year={2004},
  publisher={APS}
}

@article{makridis2020empirical,
  doi = {10.1177/0361198120911047},
  title={Empirical study on the properties of adaptive cruise control systems and their impact on traffic flow and string stability},
  author={Makridis, Michail and Mattas, Konstantinos and Ciuffo, Biagio and Re, Fabrizio and Kriston, Akos and Minarini, Fabrizio and Rognelund, Greger},
  journal={Transportation Research Record},
  volume={2674},
  number={4},
  pages={471--484},
  year={2020},
  publisher={SAGE Publications Sage CA: Los Angeles, CA}
}

@article{wilson2011car,
  doi = {10.1080/03081060.2011.530826},
  title={Car-following models: fifty years of linear stability analysis -- a mathematical perspective},
  author={Wilson, R Eddie and Ward, Jonathan A},
  journal={Transportation Planning and Technology},
  volume={34},
  number={1},
  pages={3--18},
  year={2011},
  publisher={Taylor \& Francis}
}

@inproceedings{de2004design,
  doi = {10.1109/ivs.2004.1336415},
  title={Design and test of a cooperative adaptive cruise control system},
  author={De Bruin, Dik and Kroon, Joris and Van Klaveren, Richard and Nelisse, Martin},
  booktitle={IEEE Intelligent Vehicles Symposium, 2004},
  pages={392--396},
  year={2004},
  organization={IEEE}
}

@article{milanes2013cooperative,
  doi = {10.1109/tits.2013.2278494},
  title={Cooperative adaptive cruise control in real traffic situations},
  author={Milan{\'e}s, Vicente and Shladover, Steven E and Spring, John and Nowakowski, Christopher and Kawazoe, Hiroshi and Nakamura, Masahide},
  journal={IEEE Transactions on Intelligent Transportation Systems},
  volume={15},
  number={1},
  pages={296--305},
  year = {2014},
  publisher={IEEE}
}

@article{treiber2000congested,
  doi = {10.1103/physreve.62.1805},
  title={Congested traffic states in empirical observations and microscopic simulations},
  author={Treiber, Martin and Hennecke, Ansgar and Helbing, Dirk},
  journal={Physical Review E},
  volume={62},
  number={2},
  pages={1805--1824},
  year={2000},
  publisher={APS}
}

@misc{SimBufferingStrategy,
  title = {Online simulation module: {Stop-and-Go Mitigation via Cooperative Buffering}},
  howpublished = {\url{https://www.vzu.uni-wuppertal.de/fileadmin/site/vzu/Stop-and-Go_Mitigation_via_Cooperative_Buffering.html?speed=0.4}},
  note = {Accessed: 2025-11-03}
}

@techreport{ISO15622:2018,
  title        = {{ISO 15622:2018} Intelligent transport systems --- Adaptive cruise control systems --- Performance requirements and test procedures},
  institution  = {International Organization for Standardization},
  author       = {{International Organization for Standardization}},
  year         = {2018},
  edition      = {3rd},
  address      = {Geneva, Switzerland},
  url          = {https://www.iso.org/standard/71515.html}
}

@article{Chandler1958,
 author = {Chandler,R.E. and Herman,R. and Montroll,E.W.},
 journal = {Op. Res.},
 number = {2},
 pages = {165--184},
 title = {Traffic Dynamics: Studies in Car Following},
 volume = {6},
 year = {1958}
}

@article{lhe2019adaptive,
  doi = {10.1016/j.ifacol.2019.09.004},
  title={Adaptive cruise control strategies implemented on experimental vehicles: A review},
  author={He, Yinglong and Ciuffo, Biagio and Zhou, Quan and Makridis, Michail and Mattas, Konstantinos and Li, Ji and Li, Ziyang and Yan, Fuwu and Xu, Hongming},
  journal={IFAC-PapersOnLine},
  volume={52},
  number={5},
  pages={21--27},
  year={2019},
  publisher={Elsevier}
}

@book{treiber2012traffic,
  doi = {10.1007/978-3-642-32460-4},
  title={Traffic flow dynamics: data, models and simulation},
  author={Treiber, Martin and Kesting, Arne},
  year = {2013},
  publisher={Springer Science \& Business Media}
}

@article{ploeg2013controller,
  doi = {10.1109/tits.2013.2291493},
  title={Controller synthesis for string stability of vehicle platoons},
  author={Ploeg, Jeroen and Shukla, Dipan P and Van De Wouw, Nathan and Nijmeijer, Henk},
  journal={IEEE Transactions on Intelligent Transportation Systems},
  volume={15},
  number={2},
  pages={854--865},
  year = {2014},
  publisher={IEEE}
}

@article{zheng2015stability,
  doi = {10.1109/tits.2015.2402153},
  title={Stability and scalability of homogeneous vehicular platoon: Study on the influence of information flow topologies},
  author={Zheng, Yang and Li, Shengbo Eben and Wang, Jianqiang and Cao, Dongpu and Li, Keqiang},
  journal={IEEE Transactions on Intelligent Transportation Systems},
  volume={17},
  number={1},
  pages={14--26},
  year = {2016},
  publisher={IEEE}
}

@article{orosz2016connected,
  doi = {10.1080/00423114.2016.1193209},
  title={Connected cruise control: modelling, delay effects, and nonlinear behaviour},
  author={Orosz, G{\'a}bor},
  journal={Vehicle System Dynamics},
  volume={54},
  number={8},
  pages={1147--1176},
  year={2016},
  publisher={Taylor \& Francis}
}

@article{korbmacher2025emergent,
  title={Emergent Cooperative Driving Strategies for Stop-and-Go Wave Mitigation via Multi-Agent Reinforcement Learning},
  author={Korbmacher, Raphael and Straub, Daniel and Tordeux, Antoine and Totzeck, Claudia},
  journal={arXiv preprint arXiv:2511.14378},
  year={2025}
}

\end{document}